\documentclass{article}

\usepackage{geometry}
\usepackage{ae,lmodern}
\usepackage[english]{babel}
\usepackage[utf8]{inputenc}
\usepackage[T1]{fontenc}
\usepackage{graphicx}
\usepackage{amsmath}
\usepackage{amssymb}
\usepackage{amsthm}
\usepackage{enumitem}
\usepackage{mathtools}
\usepackage{mathpartir}
\usepackage{xspace}
\usepackage{bussproofs}
\usepackage{hyperref}
\usepackage{multicol}

\newtheorem{Theorem}{Theorem}
\newtheorem{Proposition}{Proposition}
\newtheorem{Definition}{Definition}
\newtheorem{Corollary}{Corollary}
\newtheorem{Lemma}{Lemma}

\def\imp{\mathbin{\Rightarrow}}
\def\conj{\mathbin{\wedge}}
\def\disj{\mathbin{\vee}}
\def\ex{{\exists}}
\def\fa{{\forall}}

\def\eqi{\Leftrightarrow}
\def\ra{\rightarrow}
\def\lra{\hookrightarrow}
\def\o{o}

\def\assign{\mathrel{\mathop:}=}

\def\V{\mathcal{V}}
\def\C{\mathcal{C}}

\def\e{\mathfrak{e}}
\def\f{\mathfrak{f}}
\def\p{\mathfrak{p}}

\def\fv{\text{\textit{FV}}}

\def\dneq{\text{\textsc{DNE}}^{\e}}
\def\dns{\text{\textsc{DNS}}}

\usepackage{authblk}

\begin{document}

\title{Kuroda's Translation for Higher-Order Logic}
\author{Thomas Traversié}
\date{}
\renewcommand\Affilfont{\itshape\small}
\affil{Université Paris-Saclay, CentraleSupélec, MICS \\ Gif-sur-Yvette, France}
\affil{Université Paris-Saclay, Inria, CNRS, ENS-Paris-Saclay, LMF \\ Gif-sur-Yvette, France}
\renewcommand\Affilfont{\small}
\affil{\texttt{thomas.traversie@centralesupelec.fr}}

\maketitle

\begin{abstract}
Kuroda's translation embeds first-order classical logic into intuitionistic logic, such that a formula and its translation are equivalent in classical logic.
Recently, Brown and Rizkallah extended this translation to higher-order logic.
However, they showed that the translation fails in the presence of functional extensionality, and they did not prove the classical equivalence between a formula and its translation.

In this paper, we emphasize different conditions under which Kuroda's translation works in the presence of functional extensionality, including the double-negation shift.
We show that the classical equivalence between a formula and its translation does not necessarily hold in higher-order logic. 
However, it is sufficient to assume both functional extensionality and propositional extensionality. 
\end{abstract}

\maketitle

\section{Introduction}

The principle of excluded middle $A \disj \neg A$, or equivalently the double-negation elimination $\neg\neg A \imp A$, are classical logic postulates which are not derivable in intuitionistic logic.
It follows that provability in classical logic does not generally imply provability in intuitionistic logic. 
It is however possible to embed classical logic into intuitionistic logic.
Glivenko~\cite{glivenko1928} proved that, for any provable formula $A$ in classical propositional logic, there exists an intuitionistic proof of its double negation $\neg\neg A$. 
Kolmogorov~\cite{kolmogorov1925}, Gödel~\cite{godel1933}, Gentzen~\cite{gentzen1936}, and Kuroda~\cite{kuroda1951} developed translations $A \mapsto A^{\star}$ that insert double negations inside formulas. 
For any first-order context $\Gamma$ and formula $A$, such translations satisfy two properties: 
\begin{enumerate}[label=(\roman*)]
\item if $\Gamma \vdash_c A$ then $\Gamma^{\star} \vdash_i A^{\star}$ (soundness),
\item $\vdash_c A^{\star} \eqi A$ (characterization).
\end{enumerate}
Both properties~\cite{troelstra_vandalen,gaspar,kuroda_j} are essential to fully characterize translations that embed classical logic into intuitionistic logic. 
The soundness property ensures that judgments are translated from classical logic to intuitionistic logic. 
The characterization property ensures that the meaning of formulas is preserved by the translation---this is captured by checking that $A$ and $A^{\star}$ are equivalent in classical logic.
The characterization property is sometimes replaced~\cite{troelstra_vandalen,various_neg_trans} by:
\begin{itemize}
\item[(iii)] if $\Gamma^{\star} \vdash_i A^{\star}$ then $\Gamma \vdash_c A$ (reverse translation). 
\end{itemize}
The reverse translation property ensures that if a translated judgment is provable in intuitionistic logic, then the original judgment is provable in classical logic.
It follows from the characterization property.

Brown and Rizkallah~\cite{brown_rizkallah} recently investigated double-negation translations for higher-order logic. 
They showed that Kolmogorov's translation and the Gödel-Gentzen translation cannot be directly extended to higher-order logic, as they do not preserve $\beta$-conversion. 
They proved that Kuroda's translation can be extended to higher-order logic so that it satisfies the soundness property, but they did not examine the characterization and reverse translation properties. 
Moreover, they looked into different forms of extensionality, and showed that the soundness property fails in the presence of functional extensionality.

In this paper, we investigate further the extension of Kuroda's translation to higher-order logic. Our contribution is twofold.

First, we prove that the soundness property holds in the presence of functional extensionality, assuming a weak form of functional extensionality. 
We discuss alternative conditions, including the double-negation shift.

Second, we show that the characterization property and the reverse translation property---which are straightforward in first-order logic---do not necessarily hold in higher-order logic. 
We identify a necessary and sufficient condition for the characterization property.
We prove that both properties hold under functional extensionality and propositional extensionality.

The basics of higher-order logic are recalled in Section~\ref{sec_hol}. We extend Kuroda's translation to higher-order logic in Section~\ref{sec_kuroda_hol}. 
We investigate the soundness property in Section~\ref{sec_soundness}, the characterization property in Section~\ref{sec_charac}, and the reverse translation property in Section~\ref{sec_rev}.

\section{Higher-order logic}
\label{sec_hol}

In higher-order logic, functions, predicates and propositions are all terms of Church's simple type theory~\cite{church}. 
Types are defined inductively: $\iota$ is the type of individuals, $\o$ is the type of propositions, and if $\tau$ and $\sigma$ are types then $\tau \ra \sigma$ is a type. 
For every type $\tau$, let $\V_{\tau}$ be the set of variables of type $\tau$ and $\C_{\tau}$ be a set of constants of type $\tau$. 
The set of variables $\V \assign \bigcup_{\tau} \V_{\tau}$ and the set of constants $\C \assign \bigcup_{\tau} \C_{\tau}$ are assumed to be disjoint. 
For any set of constants $\C$, the sets $\Lambda^{\C}_{\tau}$ of terms of type $\tau$ are defined by induction:
\begin{itemize}
    \item For every $x \in \V_{\tau}$, $x \in \Lambda^{\C}_{\tau}$.   
    \item For every $c \in \C_{\tau}$, $c \in \Lambda^{\C}_{\tau}$.
    \item For every $x \in \V_{\tau}$ and $t \in \Lambda^{\C}_{\sigma}$, then $(\lambda x. t) \in \Lambda^{\C}_{\tau \ra \sigma}$.
    \item For every $t \in \Lambda^{\C}_{\tau \ra \sigma}$ and $u \in \Lambda^{\C}_{\tau}$, then $(t u) \in \Lambda^{\C}_{\sigma}$.
\end{itemize}
$\lambda x. t$ is a $\lambda$-abstraction and $t u$ is an application. Computation is introduced in this $\lambda$-calculus thanks to the $\beta$-reduction rule $(\lambda x. t) u \lra t[x \leftarrow u]$, where $t[x \leftarrow u]$ corresponds to the term $t$ in which $x$ has been substituted by $u$. We denote $\equiv_{\beta}$ the congruence generated by $\beta$-reduction.

Formulas are terms of type $\o$. 
There are  particular constants defining the logical connectives and quantifiers: tautology $\top$ and contradiction $\bot$ of type $\o$, negation $\neg$ of type $\o \ra \o$, implication $\imp$, conjunction $\conj$ and disjunction $\disj$ of type $\o \ra \o \ra \o$, and quantifiers $\fa_{\tau}$ and $\ex_{\tau}$ of type $(\tau \ra \o) \ra \o$. 
For convenience, terms of the form $\fa_{\tau} (\lambda x. A)$ and $\ex_{\tau} (\lambda x. A)$ are simply abbreviated as $\fa x. A$ and $\ex x. A$. 
The logical biconditional $\eqi$ is defined by $A \eqi B \assign (A \imp B) \conj (B \imp A)$. 
For every type $\tau$, we define an equality symbol $=_{\tau}$ of type $\tau \ra \tau \ra \o$. The equality symbols are infix, and we write $t = u$ when there is no ambiguity on the type $\tau$.

A context $\Gamma$ is a finite sequence of formulas. 
We write $\fv(t_1, \ldots, t_n)$ for the set of free variables that occur in the terms $t_1, \ldots, t_n$. 
The natural deduction rules for \textbf{classical logic} are given in Figure~\ref{fig_rules_hol}. 
The principle of excluded middle is equivalent to the \textit{double-negation elimination} $\neg\neg A \imp A$. 
The \textit{double-negation equivalence} $\neg\neg A \eqi A$ holds in classical logic. 
The natural deduction rules for \textbf{intuitionistic logic} are those of classical logic, except the principle of excluded middle \textsc{PEM}. 
The standard weakening inference rule is admissible in intuitionistic logic. 
The natural deduction rules for equality are given in Figure~\ref{fig_rules_eq}.

We write $\Gamma \vdash_i A$ when $\Gamma \vdash A$ is derivable in intuitionistic logic, and $\Gamma \vdash_c A$ when it is derivable in classical logic. 
For any $k \in \{ c, i \}$, we write $\Gamma \vdash_k^* A$ with $* \in \{ \e, \e\p, \e\f, \e\f\p \}$ when $\Gamma \vdash_k A$ is derivable with possibly additional inference rules: with \textsc{Eq-I} and \textsc{Eq-E} if $\e$ is in $*$, with \textsc{PropExt} if $\p$ is in $*$, and with \textsc{FunExt} if $\f$ is in $*$.

\begin{figure}[tb]
\begin{mathpar}
\inferrule*[right={Imp-I}]{\Gamma, A \vdash B}{\Gamma \vdash A \imp B}

\inferrule*[right={Imp-E}]{\Gamma \vdash A \imp B \\ \Gamma \vdash A}{\Gamma \vdash B}
\end{mathpar}
\begin{mathpar}
\inferrule*[right={And-I}]{\Gamma \vdash A \\ \Gamma \vdash B}{\Gamma \vdash A \conj B}

\inferrule*[right={And-EL}]{\Gamma \vdash A \conj B}{\Gamma \vdash A}

\inferrule*[right={And-ER}]{\Gamma \vdash A \conj B}{\Gamma \vdash B}
\end{mathpar}
\begin{mathpar}
\inferrule*[right={Or-IL}]{\Gamma \vdash A}{\Gamma \vdash A \disj B}

\inferrule*[right={Or-IR}]{\Gamma \vdash B}{\Gamma \vdash A \disj B}

\inferrule*[right={Or-E}]{\Gamma \vdash A \disj B \\ \Gamma, A \vdash C \\ \Gamma, B \vdash C}{\Gamma \vdash C}
\end{mathpar}
\begin{mathpar}
\inferrule*[right={Not-I}]{\Gamma, A \vdash \bot}{\Gamma \vdash \neg A}

\inferrule*[right={Not-E}]{\Gamma \vdash \neg A \\ \Gamma \vdash A}{\Gamma \vdash \bot}
\end{mathpar}
\begin{mathpar}
\inferrule*[right={Bot-E}]{\Gamma \vdash \bot}{\Gamma \vdash A}

\inferrule*[right={Top-I}]{ }{\Gamma \vdash \top}
\end{mathpar}
\begin{mathpar}
\inferrule*[right={All-I}]{\Gamma \vdash P x \\ x \notin \fv(\Gamma,P)}{\Gamma \vdash \fa P}

\inferrule*[right={All-E}]{\Gamma \vdash \fa P}{\Gamma \vdash P t}
\end{mathpar}
\begin{mathpar}
\inferrule*[right={Ex-I}]{\Gamma \vdash P t}{\Gamma \vdash \ex P}

\inferrule*[right={Ex-E}]{\Gamma \vdash \ex P \\ \Gamma, P x \vdash A \\ x \notin \fv(\Gamma, P, A)}{\Gamma \vdash A}
\end{mathpar}
\begin{mathpar}
\inferrule*[right={Ax}]{ }{\Gamma, A, \Delta \vdash A}

\inferrule*[right={Conv}]{\Gamma \vdash A \\ A \equiv_{\beta} B}{\Gamma \vdash B}

\inferrule*[right={PEM}]{ }{\Gamma \vdash A \disj \neg A}
\end{mathpar}
\caption{Natural deduction rules for higher-order logic.}
\label{fig_rules_hol}
\end{figure}

\begin{figure}[tb]
\begin{mathpar}
\inferrule*[right={Eq-I}]{ }{\Gamma \vdash u = u}

\inferrule*[right={Eq-E}]{\Gamma \vdash P u \\ \Gamma \vdash u = v}{\Gamma \vdash P v}

\inferrule*[right={FunExt}]{\Gamma \vdash f x = g x \\ x \notin \fv(\Gamma, f, g)}{\Gamma \vdash f = g}

\inferrule*[right={PropExt}]{\Gamma \vdash A \eqi B}{\Gamma \vdash A = B}
\end{mathpar}
\caption{Natural deduction rules for equality.}
\label{fig_rules_eq}
\end{figure}
Let us recall some well-known results about intuitionistic logic.

\begin{Proposition}
\label{prop}
Let $A$ and $B$ be formulas, $P$ be a predicate, and $u$ and $v$ be two terms.
\begin{multicols}{2}
\begin{enumerate}
\item $\vdash_i \neg\neg \bot \imp \bot$ \label{prop_doublenegbot}
\item $\vdash_i \neg\neg \top \imp \top$ \label{prop_doublenegtop}
\item $\vdash_i \neg\neg (A \disj \neg A)$ \label{prop_doubleneg_middle}
\item $\vdash_i A \imp \neg\neg A$ \label{prop_doubleneg}
\item $\vdash_i \neg\neg\neg A \eqi \neg A$ \label{prop_tripleneg}
\item $\vdash_i \neg\neg (A \imp B) \eqi (\neg\neg A \imp \neg\neg B)$ \label{prop_doublenegimp}
\item $\vdash_i \neg\neg (A \conj B) \eqi (\neg\neg A \conj \neg\neg B)$ \label{prop_doublenegconj}
\item $\vdash_i \neg (A \disj B) \eqi (\neg A \conj \neg B)$ \label{prop_negdisj}
\item $\vdash_i \neg\neg \fa P \imp \fa x. \neg\neg (P x)$ \label{prop_doublenegfa}
\item $\vdash_i \neg \ex P \eqi \fa x. \neg (P x)$ \label{prop_negex}
\end{enumerate}
\end{multicols}
\end{Proposition}

\begin{proof}
These results are well-known, see for example \cite[Chapter 2]{troelstra_vandalen}.
\end{proof}

\section{Kuroda's translation for higher-order logic}
\label{sec_kuroda_hol}

Kuroda's translation for first-order logic~\cite{kuroda1951} inserts a double negation in front of formulas and a double negation after every universal quantifier. More formally, we first define $A_{Ku}$ by induction on $A$:
\[
  \begin{array}{llll}
        (A \imp B)_{Ku} \assign A_{Ku} \imp B_{Ku} &(\neg A)_{Ku} \assign \neg A_{Ku} &P_{Ku} \assign P \text{ if $P$ atomic} \\
        (A \conj B)_{Ku} \assign A_{Ku} \conj B_{Ku} &\top_{Ku} \assign \top &(\fa x. A)_{Ku} \assign \fa x. \neg\neg A_{Ku} \\
        (A \disj B)_{Ku} \assign A_{Ku} \disj B_{Ku} &\bot_{Ku} \assign \bot &(\ex x. A)_{Ku} \assign \ex x. A_{Ku} \\
  \end{array}
\]
and then we set $A^{Ku} \assign \neg\neg A_{Ku}$. We extend Kuroda's translation to the terms of higher-order logic.

\begin{Definition}[Kuroda's translation for higher-order logic] Let $A$ be a formula in higher-order logic. Its Kuroda's translation is $A^{Ku} \assign \neg\neg A_{Ku}$, where $A_{Ku}$ is inductively defined by:
\[
  \begin{array}{lll}
        x_{Ku} &\assign &x \\
        c_{Ku} &\assign &\left\{
        \begin{array}{ll}
            \lambda p. \fa x. \neg\neg (p x) &\text{ if $c = \fa$} \\
            c &\text{ otherwise} \\
        \end{array}
        \right.\\
        (\lambda x. t)_{Ku} &\assign &\lambda x. t_{Ku} \\
        (t u)_{Ku} &\assign &t_{Ku} u_{Ku} \\
  \end{array}
\]
\end{Definition}
While in first-order logic we have $(A[z \leftarrow w])^{Ku} = A^{Ku}[z \leftarrow w]$, this result cannot hold anymore in higher-order logic, since $w$ is modified when it contains universal quantifiers $\fa$. Instead, we have $(A[z \leftarrow w])^{Ku} = A^{Ku}[z \leftarrow w_{Ku}]$.

\begin{Proposition}
\label{prop_Ku_subst_term}
For any term $t$, we have $(t[z \leftarrow w])_{Ku} = t_{Ku}[z \leftarrow w_{Ku}]$.
\end{Proposition}

\begin{proof}
By induction on the term $t$.
\end{proof}

\begin{Corollary}
\label{prop_Ku_subst}
For any higher-order formula $A$, we have $(A[z \leftarrow w])^{Ku} = A^{Ku}[z \leftarrow w_{Ku}]$.
\end{Corollary}
Higher-order logic is defined using simple type theory, so $\beta$-conversions may be used in the derivations. As Kuroda's proof relies on the fact that we can translate each step of the derivation, each time we have $A \equiv_{\beta} B$ in the classical derivation, we want to use $A^{Ku} \equiv_{\beta} B^{Ku}$ in the intuitionistic derivation.

\begin{Proposition}
\label{conv_hol_term_Ku}
For any terms $t$ and $u$, if $t \equiv_{\beta} u$ then $t_{Ku} \equiv_{\beta} u_{Ku}$.
\end{Proposition}

\begin{proof}
We have $((\lambda x. t) u)_{Ku} = (\lambda x. t_{Ku}) u_{Ku} \lra t_{Ku}[x \leftarrow u_{Ku}]$, so that $((\lambda x. t) u)_{Ku} \equiv_{\beta} (t[x \leftarrow u])_{Ku}$ using Proposition~\ref{prop_Ku_subst_term}. Closure by context, reflexivity, symmetry, and transitivity are immediate.
\end{proof}

\begin{Corollary}
\label{conv_hol_Ku}
For any higher-order formulas $A$ and $B$, if $A \equiv_{\beta} B$ then $A_{Ku} \equiv_{\beta} B_{Ku}$ and $A^{Ku} \equiv_{\beta} B^{Ku}$.
\end{Corollary}

\section{Soundness property}
\label{sec_soundness}

In this section, we prove the soundness property, with and without the extensionality principles. In particular, we identify different conditions that are sufficient to derive the soundness property in the presence of functional extensionality.

Note that we state and prove the soundness property with $\Gamma_{Ku}$ instead of $\Gamma^{Ku}$, thus introducing fewer double negations in the context.
All the results can be simply adapted to have the version with $\Gamma^{Ku}$.

\subsection{Without extensionality}

To prove that $\Gamma \vdash_c A$ implies $\Gamma_{Ku} \vdash_i A^{Ku}$, we do not follow Brown and Rizkallah proof~\cite{brown_rizkallah}, which proceeds in two steps---transforming $A$ into a formula $A'$ that does not contain any universal quantifier and then applying to $A'$ an extension of Glivenko's theorem to higher-order logic without universal quantifiers. 
We proceed to a more direct proof that follows the intuition of the first-order case.

\begin{Theorem}
\label{thm_c2i}
Let $A$ be a formula and $\Gamma$ be a context in higher-order logic.
If $\Gamma \vdash_c A$ then $\Gamma_{Ku} \vdash_i A^{Ku}$.
\end{Theorem}

\begin{proof}
By induction on the derivation:
\begin{itemize}
\item \underline{\textsc{All-I}}: By induction, we have $\Gamma_{Ku} \vdash_i\neg\neg (P_{Ku} x)$. 
Using \textsc{All-I}, we derive $\Gamma_{Ku} \vdash_i\fa x. \neg\neg (P_{Ku} x)$. 
We conclude using Proposition~\ref{prop}(\ref{prop_doubleneg}).
\item \underline{\textsc{All-E}}: By induction, we have $\Gamma_{Ku} \vdash_i\neg\neg \fa x. \neg\neg (P_{Ku} x)$. 
It follows from Proposition~\ref{prop}(\ref{prop_doublenegfa}) that $\Gamma_{Ku} \vdash_i\fa x. \neg\neg\neg\neg (P_{Ku} x)$. Using \textsc{All-E}, we obtain $\Gamma_{Ku} \vdash_i\neg\neg\neg\neg (P_{Ku} t_{Ku})$. 
We conclude using Proposition~\ref{prop}(\ref{prop_tripleneg}). 
\item \underline{\textsc{Ex-I}}: By induction, we have $\Gamma_{Ku} \vdash_i\neg\neg (P_{Ku} t_{Ku})$. 
We write $\Delta$ for the context $\Gamma_{Ku}, \fa x. \neg (P_{Ku} x)$.
\begin{prooftree}
\AxiomC{Hypothesis}
\RightLabel{Weakening}
\UnaryInfC{$\Delta \vdash_i\neg\neg (P_{Ku} t_{Ku})$}
\AxiomC{}
\RightLabel{\textsc{Ax}}
\UnaryInfC{$\Delta \vdash_i\fa x. \neg (P_{Ku} x)$}
\RightLabel{\textsc{All-E}}
\UnaryInfC{$\Delta \vdash_i\neg (P_{Ku} t_{Ku})$}
\RightLabel{\textsc{Not-E}}
\BinaryInfC{$\Delta \vdash_i\bot$}
\RightLabel{\textsc{Not-I}}
\UnaryInfC{$\Gamma_{Ku} \vdash_i\neg \fa x. \neg (P_{Ku} x)$}
\RightLabel{Proposition~\ref{prop}(\ref{prop_negex})}
\UnaryInfC{$\Gamma_{Ku} \vdash_i\neg\neg \ex P_{Ku}$}
\end{prooftree}
\item \underline{\textsc{Ex-E}}: We get $\Gamma_{Ku} \vdash_i\neg\neg \ex P_{Ku}$ and $\Gamma_{Ku}, P_{Ku} x \vdash_i\neg\neg A_{Ku}$ by induction. 
We want to prove $\Gamma_{Ku} \vdash_i\neg\neg A_{Ku}$. 
We use \textsc{Not-I} and then \textsc{Not-E} with $\neg \fa x. \neg (P_{Ku} x)$. The proof of the first subgoal, where the context $\Delta$ denotes $\Gamma_{Ku}, \neg A_{Ku}$, is:
\begin{prooftree}
\AxiomC{Hypothesis}
\RightLabel{Weakening}
\UnaryInfC{$\Delta, P_{Ku} x \vdash_i\neg\neg A_{Ku}$}
\AxiomC{}
\RightLabel{\textsc{Ax}}
\UnaryInfC{$\Delta, P_{Ku} x \vdash_i\neg A_{Ku}$}
\RightLabel{\textsc{Not-E}}
\BinaryInfC{$\Delta, P_{Ku} x \vdash_i\bot$}
\RightLabel{\textsc{Not-I}}
\UnaryInfC{$\Delta \vdash_i\neg (P_{Ku} x)$}
\RightLabel{\textsc{All-I}}
\UnaryInfC{$\Delta \vdash_i\fa x. \neg (P_{Ku} x)$}
\RightLabel{Proposition~\ref{prop}(\ref{prop_doubleneg})}
\UnaryInfC{$\Delta \vdash_i\neg\neg \fa x. \neg (P_{Ku} x)$}
\end{prooftree}
The second subgoal $\Gamma_{Ku}, \neg A_{Ku} \vdash_i\neg \fa x. \neg (P_{Ku} x)$ derives from Proposition~\ref{prop}(\ref{prop_negex}), weakening and the first induction hypothesis.
\item \underline{\textsc{Or-IL}}: By induction, we have $\Gamma_{Ku} \vdash_i\neg\neg A_{Ku}$. 
We want to prove that $\Gamma_{Ku} \vdash_i\neg\neg (A_{Ku} \disj B_{Ku})$. 
We apply \textsc{Not-I} and then \textsc{Not-E} with $\neg A_{Ku}$. The first subgoal is $\Gamma_{Ku}, \neg (A_{Ku} \disj B_{Ku}) \vdash_i\neg\neg A_{Ku}$. 
It is proved by weakening the induction hypothesis. The second subgoal is:

\begin{prooftree}
\AxiomC{}
\RightLabel{\textsc{Ax}}
\UnaryInfC{$\Gamma_{Ku}, \neg (A_{Ku} \disj B_{Ku}) \vdash_i\neg (A_{Ku} \disj B_{Ku})$}
\RightLabel{Proposition~\ref{prop}(\ref{prop_negdisj})}
\UnaryInfC{$\Gamma_{Ku}, \neg (A_{Ku} \disj B_{Ku}) \vdash_i\neg A_{Ku} \conj \neg B_{Ku}$}
\RightLabel{\textsc{And-EL}}
\UnaryInfC{$\Gamma_{Ku}, \neg (A_{Ku} \disj B_{Ku}) \vdash_i\neg A_{Ku}$}
\end{prooftree}
We proceed similarly for \underline{\textsc{Or-IR}}.
\item \underline{\textsc{Or-E}}: We have $\Gamma_{Ku} \vdash_i\neg\neg (A_{Ku} \disj B_{Ku})$, and $\Gamma_{Ku}, A_{Ku} \vdash_i\neg\neg C_{Ku}$, and $\Gamma_{Ku}, B_{Ku} \vdash_i\neg\neg C_{Ku}$ by induction. 
We want to prove $\Gamma_{Ku} \vdash_i\neg\neg C_{Ku}$. 
We use \textsc{Not-I} and \textsc{Not-E} with $\neg A_{Ku} \conj \neg B_{Ku}$.

For the first subgoal, we have to prove $\Gamma_{Ku}, \neg C_{Ku} \vdash_i\neg (\neg A_{Ku} \conj \neg B_{Ku})$. 
Using Proposition~\ref{prop}(\ref{prop_negdisj}), we have to show $\Gamma_{Ku}, \neg C_{Ku} \vdash_i\neg\neg (A_{Ku} \disj B_{Ku})$. 
We prove it by weakening the first induction hypothesis.

For the second subgoal, we have to prove $\Gamma_{Ku}, \neg C_{Ku} \vdash_i\neg A_{Ku} \conj \neg B_{Ku}$.
We prove both $\Gamma_{Ku}, \neg C_{Ku} \vdash_i\neg A_{Ku}$ and $\Gamma_{Ku}, \neg C_{Ku} \vdash_i\neg B_{Ku}$. 
Let us show $\Gamma_{Ku}, \neg C_{Ku} \vdash_i\neg A_{Ku}$. We use \textsc{Not-I} and \textsc{Not-E} with $\neg C_{Ku}$. 
We derive $\Gamma_{Ku}, \neg C_{Ku}, A_{Ku} \vdash_i\neg\neg C_{Ku}$ by weakening the third induction hypothesis and $\Gamma_{Ku}, \neg C_{Ku}, A_{Ku} \vdash_i\neg C_{Ku}$ by \textsc{Ax}.
We show $\Gamma_{Ku}, \neg C_{Ku} \vdash_i\neg B_{Ku}$ similarly. 
\item \underline{\textsc{Conv}} derives from Corollary~\ref{conv_hol_Ku}.
\item The remaining items are direct applications of Proposition~\ref{prop}.
\end{itemize}
\end{proof}

\subsection{With extensionality}

We now investigate the soundness property in the setting of extensionality. We want to prove that $\Gamma \vdash_c^* A$ implies $\Gamma_{Ku} \vdash_i^* A^{Ku}$ whatever $* \in \{ \e, \e\p, \e\f, \e\f\p \}$. Brown and Rizkallah showed that it fails in the presence of functional extensionality. To get around this problem, it is sufficient to assume $\fa f \fa g. (\fa x. \neg\neg (f x = g x)) \imp \neg\neg (f= g)$, written $\Delta^{\f}$.

\begin{Theorem}
Let $A$ be a formula and $\Gamma$ be a context in higher-order logic.
\begin{enumerate}
\item For $* \in \{ \e, \e\p \}$, if $\Gamma \vdash_c^* A$ then $\Gamma_{Ku} \vdash_i^* A^{Ku}$.
\item For $* \in \{ \e\f, \e\f\p \}$, if $\Gamma \vdash_c^* A$ then $\Delta^{\f}, \Gamma_{Ku} \vdash_i^* A^{Ku}$.
\end{enumerate}
\end{Theorem}

\begin{proof}
For the first item, we complete the proof of Theorem~\ref{thm_c2i}: 
\begin{itemize}
\item \underline{\textsc{Eq-I}} derives from \textsc{Eq-I} and Proposition~\ref{prop}(\ref{prop_doubleneg}).
\item \underline{\textsc{Eq-E}}: We have $\Gamma_{Ku} \vdash_i^{\e\p} \neg\neg (P_{Ku} u_{Ku})$ and $\Gamma_{Ku} \vdash_i^{\e\p} \neg\neg (u_{Ku} = v_{Ku})$ by induction. We directly have $\Gamma_{Ku} \vdash_i^{\e\p} P_{Ku} u_{Ku} \imp (u_{Ku} = v_{Ku}) \imp P_{Ku} v_{Ku}$ by \textsc{Eq-E}. We get $\Gamma_{Ku} \vdash_i^{\e\p} \neg\neg (P_{Ku} u_{Ku}) \imp \neg\neg (u_{Ku} = v_{Ku}) \imp \neg\neg (P_{Ku} v_{Ku})$ using Proposition~\ref{prop}(\ref{prop_doubleneg}) and Proposition~\ref{prop}(\ref{prop_doublenegimp}). Therefore, we conclude $\Gamma_{Ku} \vdash_i^{\e\p} \neg\neg (P_{Ku} v_{Ku})$.
\item \underline{\textsc{PropExt}}: By induction, we have $\Gamma_{Ku} \vdash_i^{\e\p} \neg\neg (A_{Ku} \eqi B_{Ku})$. Using \textsc{PropExt}, we directly have $\Gamma_{Ku} \vdash_i^{\e\p} (A_{Ku} \eqi B_{Ku}) \imp (A_{Ku} = B_{Ku})$. We get $\Gamma_{Ku} \vdash_i^{\e\p} \neg\neg (A_{Ku} \eqi B_{Ku}) \imp \neg\neg (A_{Ku} = B_{Ku})$ by Proposition~\ref{prop}(\ref{prop_doubleneg}) and Proposition~\ref{prop}(\ref{prop_doublenegimp}). We conclude $\Gamma_{Ku} \vdash_i^{\e\p} \neg\neg (A_{Ku} = B_{Ku})$.
\end{itemize}
For the second item, we adapt the three cases of the first item and we prove the additional \underline{\textsc{FunExt}} case. We use the assumption $\Delta^{\f}$ on the induction hypothesis $\Delta^{\f}, \Gamma_{Ku} \vdash_i \neg\neg (f_{Ku} x = g_{Ku} x)$ to get $\Delta^{\f}, \Gamma_{Ku} \vdash_i \neg\neg (f_{Ku} = g_{Ku})$.
\end{proof}

The formula $\Delta^{\f}$ corresponds to a weak form of functional extensionality, where double negations have been inserted in front of equality predicates. In classical logic, it is equivalent to functional extensionality. We link $\Delta^{\f}$ to the well-known double-negation shift $(\fa x. \neg\neg A) \imp (\neg\neg \fa x. A)$, written $\dns$. Such a principle implies $\Delta^{\f}$ in intuitionistic logic with functional extensionality. 

\begin{Proposition}
$\dns \vdash_i^{\e\f} \Delta^{\f}$.
\end{Proposition}

\begin{proof}
We have $\dns \vdash_i^{\e\f} \fa f \fa g. (\fa x. f x = g x) \imp (f = g)$ using \textsc{FunExt}. We successively apply Propositions~\ref{prop}(\ref{prop_doubleneg}), \ref{prop}(\ref{prop_doublenegfa}) and \ref{prop}(\ref{prop_doublenegimp}) to derive $\dns \vdash_i^{\e\f} \fa f \fa g. \neg\neg (\fa x. f x = g x) \imp \neg\neg (f = g)$. Using $\dns$, we conclude $\dns \vdash_i^{\e\f} \Delta^{\f}$.
\end{proof}

The double-negation shift is derivable in classical logic but not in intuitionistic logic. It is deeply connected to Glivenko's theorem---which is generalized by Kuroda's translation. More precisely, such a principle is required to extend Glivenko's theorem to first-order logic~\cite{umezawa,gabbay} and to substructural first-order logic~\cite{substructural}. Moreover, the modal counterpart of the double-negation shift $\square\neg\neg A \imp \neg\neg \square A$ plays an important role when studying Glivenko's theorem for modal logic~\cite{litak}. Here, we have emphasized a new connection between the double-negation shift and Kuroda's translation.

To prove the soundness property in the presence of functional extensionality, it is therefore sufficient to assume the weak form of functional extensionality or the double-negation shift. 

We could also assume the double-negation elimination on equality predicates $\fa x \fa y. \neg\neg (x = y) \imp x = y$, written $\dneq$. Indeed, it implies $\Delta^{\f}$ in intuitionistic logic with functional extensionality. 
However, such a principle must be handled with care: when combined with the comprehension axiom of set theory or with propositional extensionality, it entails the double-negation elimination. 
For instance, let us prove $\dneq \vdash_i^{\e\p} \neg\neg P \imp P$ for any proposition $P$. We easily derive $\vdash_i^{\e\p} \neg\neg P \imp \neg\neg (P = \top)$ and $\vdash_i^{\e} (P = \top) \imp P$.
The double-negation elimination then follows from $\dneq$ applied to $P$ and $\top$.

Another approach to get around the problem of the soundness property in the presence of functional extensionality is to apply a translation that eliminates extensionality~\cite{gandy,takeuti,elim_funext}. To do so, the Takeuti-Gandy translation~\cite{gandy,takeuti} resorts to relativization and quantifies over extensional elements. This translation was later refined~\cite{elim_funext} and coupled with a double-negation translation. 
The main advantage of our approach is that the translated statement remains close to the original statement, as no predicates have been inserted to account for the elimination of functional extensionality. In exchange, however, we assume an additional axiom.

\section{Characterization property}
\label{sec_charac}

We have shown that, if the formula $A$ is provable from $\Gamma$, then there exists an intuitionistic proof of $A^{Ku}$ from $\Gamma_{Ku}$. 
We now want to prove the characterization property: $A$ and $A^{Ku}$ are classically equivalent. 
Such a result is routine in first-order logic, but does not generally hold in higher-order logic. 

\subsection{Counter-example}
\label{sec_counterex}

Consider a constant $P$ of type $\o \ra \o$ and a proposition $A$. 
We would like to show $\vdash_c (P A)^{Ku} \eqi P A$, that is $\vdash_c \neg\neg (P A_{Ku}) \eqi P A$. 
As we are in classical logic, it is equivalent to derive $\vdash_c P A_{Ku} \eqi P A$. 
It is impossible to show this without further assumptions on the predicate $P$, even when $A_{Ku}$ and $A$ are equivalent. 
Here, the problem is that the translation may insert double negations inside the argument of the predicate $P$, but we cannot replace $A_{Ku}$ by $A$ in the absence of extensionality.

This issue is not the only consequence of the lack of extensionality in higher-order logic.
For example, $\top$ and $\top \conj \top$ are equivalent, but we cannot prove that $P \top$ and $P (\top \conj \top)$ are equivalent without further assumptions.
This can be shown using the \textit{intensional} models of higher-order logic developed by Takahashi~\cite{takahashi} and Prawitz~\cite{prawitz_hol}, based on the V-complex construction~\cite{andrews_resolution}.
In such models, the denotation of a proposition is not its truth value, but a pair consisting of its truth value and a syntactic value.
$\top$ and $\top \conj \top$ have the same truth value, but different syntactic values.
The truth value of $P (\top \conj \top)$ (respectively $P \top$) depends on the whole denotation of $\top \conj \top$ (respectively $\top$).
It follows that $P \top$ and $P (\top \conj \top)$ can have different truth values.
Thus, in the absence of extensionality, we cannot prove that $P \top$ and $P (\top \conj \top)$ are equivalent, even though $\top$ and $\top \conj \top$ are.
The same reasoning applies for $P A_{Ku}$ and $P A$.

\subsection{Necessary and sufficient condition}

When trying to adapt the proof of the characterization property from first-order logic to higher-order logic, we identify a necessary and sufficient condition, called Kuroda-equivalence.

\begin{Definition}[Kuroda-equivalence]
For any constant or variable $t$ of type $\tau_1 \ra \ldots \ra \tau_n \ra \o$ and any terms $u_1, \ldots, u_n$ of type $\tau_1, \ldots, \tau_n$, we have $\vdash_c (t u_1 \ldots u_n)_{Ku} \eqi t u_1 \ldots u_n$.
\end{Definition}

\begin{Theorem}
\label{thm_ku_eq}
The characterization property holds if and only if the Kuroda-equivalence holds.
\end{Theorem}

\begin{proof}
For the direct implication, we apply the characterization property to the formula $t u_1 \ldots u_n$ to get $\vdash_c (t u_1 \ldots u_n)^{Ku} \eqi t u_1 \ldots u_n$, so we derive $\vdash_c (t u_1 \ldots u_n)_{Ku} \eqi t u_1 \ldots u_n$.

For the reverse implication, we only have to show $\vdash_c A_{Ku} \eqi A$ as we are in classical logic. 
Let $A_{\beta}$ be the $\beta$-normal form of $A$. 
We know that $A_{\beta} \equiv_{\beta} A$, so we derive $(A_{\beta})_{Ku} \equiv_{\beta} A_{Ku}$ by Corollary~\ref{conv_hol_Ku}. 
We do a proof by cases on $A_{\beta}$. 
If $A_{\beta}$ is a variable or a constant, $(A_{\beta})_{Ku}$ corresponds to $A_{\beta}$, therefore we directly have $\vdash_c (A_{\beta})_{Ku} \eqi A_{\beta}$ and we conclude $\vdash_c A_{Ku} \eqi A$.
Otherwise, $A_{\beta}$ is an application $t u_1 \ldots u_n$, where $t$ is a variable or a constant of type $\tau_1 \ra \ldots \ra \tau_n \ra \o$ and $u_1, \ldots, u_n$ are terms of type $\tau_1, \ldots, \tau_n$.
By the Kuroda-equivalence, we have $\vdash_c (A_{\beta})_{Ku} \eqi A_{\beta}$, hence we conclude $\vdash_c A_{Ku} \eqi A$.
\end{proof}
 
The counter-example~\ref{sec_counterex} precisely shows that the Kuroda-equivalence does not necessarily hold. 
To have the characterization property, we could impose the Kuroda-equivalence as a principle. However, such a condition is difficult to enforce, as it not only applies to the logical connectives and quantifiers, but also to any variable or constant representing a predicate.

\subsection{Extensionality conditions}

We now investigate the characterization property in the setting of extensionality. When assuming functional extensionality and propositional extensionality, the formulas $A$ and $A^{Ku}$ are equal, and therefore classically equivalent.

\begin{Lemma}
\label{prop_term_Ku}
For any term $t$, we have $\vdash_c^{\e\f\p} t_{Ku} = t$.
\end{Lemma}

\begin{proof}
We proceed by induction on the term $t$:
\begin{itemize}
\item \underline{Variables and constants}: Using \textsc{Eq-I}, we directly derive $\vdash_c^{\e\f\p} x_{Ku} = x$ and $\vdash_c^{\e\f\p} c_{Ku} = c$ for $c \neq \fa$.  
We prove $\vdash_c^{\e\f\p} \fa_{Ku} = \fa$ with
\begin{prooftree}
\AxiomC{}
\RightLabel{\textsc{Eq-I}}
\UnaryInfC{$\vdash_c^{\e\f\p} \fa q = \fa q$}
\AxiomC{$\vdash_c^{\e\f\p} q y \eqi \neg\neg (q y)$}
\RightLabel{\textsc{PropExt}}
\UnaryInfC{$\vdash_c^{\e\f\p} q y = \neg\neg (q y)$}
\RightLabel{\textsc{FunExt} and \textsc{Conv}}
\UnaryInfC{$\vdash_c^{\e\f\p} q = (\lambda x. \neg\neg (q x))$}
\RightLabel{\textsc{Eq-E} and \textsc{Conv}}
\BinaryInfC{$\vdash_c^{\e\f\p} \fa (\lambda x. \neg\neg (q x)) = \fa q$}
\RightLabel{\textsc{FunExt} and \textsc{Conv}}
\UnaryInfC{$\vdash_c^{\e\f\p} (\lambda p. \fa (\lambda x. \neg\neg (p x))) = \fa$}
\end{prooftree}
as we have $\vdash_c^{\e\f\p} q y \eqi \neg\neg (q y)$ in classical logic.
\item \underline{Applications}: We have $\vdash_c^{\e\f\p} t_{Ku} u_{Ku} = t_{Ku} u_{Ku}$ using \textsc{Eq-I}. Using successively \textsc{Eq-E} with the induction hypotheses $\vdash_c^{\e\f\p} t_{Ku} = t$ and $\vdash_c^{\e\f\p} u_{Ku} = u$, we obtain $\vdash_c^{\e\f\p} t_{Ku} u_{Ku} = t u$.
\item \underline{Abstractions}: By induction hypothesis, we have $t_{Ku} = t$. Using \textsc{FunExt} and \textsc{Conv}, we derive $\vdash_c^{\e\f\p} (\lambda x. t)_{Ku} = \lambda x. t$.
\end{itemize}
\end{proof}

\begin{Theorem}
\label{prop_Ku_eqi}
For any higher-order formula $A$, we have $\vdash_c^{\e\f\p} A^{Ku} \eqi A$.
\end{Theorem}

\begin{proof}
By Lemma~\ref{prop_term_Ku}, $\vdash_c^{\e\f\p} A_{Ku} = A$. Using \textsc{Eq-E} on the tautology $A \eqi A$, we get $\vdash_c^{\e\f\p} A_{Ku} \eqi A$. We conclude $\vdash_c^{\e\f\p} A^{Ku} \eqi A$ as we are in classical logic.
\end{proof}

Note that Lemma~\ref{prop_term_Ku} imply the Kuroda-equivalence, which provides an alternative proof of Theorem~\ref{prop_Ku_eqi}.

\section{Reverse translation property}
\label{sec_rev}

The characterization property does not generally hold in higher-order logic.
In this section, we show that the reverse translation property---a weaker property---does not hold either. 

As in Section~\ref{sec_soundness}, we state the reverse translation property with $\Gamma_{Ku}$ instead of $\Gamma^{Ku}$.
Again, all the results can be simply adapted to have the version with $\Gamma^{Ku}$.

\subsection{Counter-example}

To prove that the reverse translation property does not hold in higher-order logic, we exhibit some $\Gamma$ and $A$ such that $\Gamma_{Ku} \vdash_i A^{Ku}$ holds and $\Gamma \vdash_c A$ does not. 
Consider a constant R of type $\o \ra \o \ra \o$ and two predicates $P$ and $P'$ of type $\tau \ra \o$.
We take
\[
  \begin{array}{lll}
        \Gamma &\assign &\fa q. R (q (\lambda x. \neg\neg (P x))) (q (\lambda x. \neg\neg (P' x))) \\
        A &\assign &R (\fa (\lambda x. P x)) (\fa (\lambda x. P' x)) \\
  \end{array}
\]
and we therefore have
\[
  \begin{array}{lll}
        \Gamma_{Ku} &\equiv_{\beta} &\fa q. \neg\neg (R (q (\lambda x. \neg\neg (P x))) (q (\lambda x. \neg\neg (P' x)))) \\
        A^{Ku} &\equiv_{\beta} &\neg\neg (R (\fa (\lambda x. \neg\neg (P x))) (\fa (\lambda x. \neg\neg (P' x)))) \\
  \end{array}
\]
In intuitionistic logic, we derive $\Gamma_{Ku} \vdash_i A^{Ku}$ using \textsc{All-E} with $\fa$.

In classical logic, however, we cannot derive $\Gamma \vdash_c A$.
Indeed, we cannot apply \textsc{All-E} with $\fa$ anymore: the extra double negations inside the arguments of $R$ cannot be removed without further assumptions.
We cannot apply \textsc{All-E} with $\lambda y. \fa (\lambda x. P x)$ either, as we would obtain $R (\fa (\lambda x. P x)) (\fa (\lambda x. P x))$ instead of $R (\fa (\lambda x. P x)) (\fa (\lambda x. P' x))$.
As well as for the characterization property, it is the absence of extensionality that allows us to build a counter-example.

\subsection{Sufficient conditions}

Since the characterization property implies the reverse translation property, the sufficient conditions investigated in Section~\ref{sec_charac} apply. 
The reverse translation property holds under the Kuroda-equivalence.

\begin{Theorem}
If the Kuroda-equivalence holds, then the reverse translation property holds.
\end{Theorem}

\begin{proof}
Suppose $\Gamma_{Ku} \vdash_i A^{Ku}$.
We directly have $\Gamma_{Ku} \vdash_c A^{Ku}$, and we get $\Gamma^{Ku} \vdash_c A^{Ku}$ using the double-negation elimination.
We use the characterization property, that holds under Kuroda-equivalence by Theorem~\ref{thm_ku_eq}.
\end{proof}

The reverse translation property also holds under functional extensionality and propositional extensionality.

\begin{Theorem}
Let $A$ be a formula and $\Gamma$ be a context in higher-order logic.
\begin{enumerate}
\item If $\Gamma_{Ku} \vdash_i A^{Ku}$ then $\Gamma \vdash_c^{\e\f\p} A$.
\item For $* \in \{ \e, \e\p \}$, if $\Gamma_{Ku} \vdash_i^* A^{Ku}$ then $\Gamma \vdash_c^{\e\f\p} A$.
\item For $* \in \{ \e\f, \e\f\p \}$, if $\Delta^{\f}, \Gamma_{Ku} \vdash_i^* A^{Ku}$ then $\Gamma \vdash_c^{\e\f\p} A$.
\end{enumerate}
\end{Theorem}

\begin{proof}
For the first and second item, we directly have $\Gamma_{Ku} \vdash_c^{\e\f\p} A^{Ku}$. 
We get $\Gamma^{Ku} \vdash_c^{\e\f\p} A^{Ku}$ using the double-negation elimination, and then we derive $\Gamma \vdash_c^{\e\f\p} A$ using Theorem~\ref{prop_Ku_eqi}. 
For the third item, we similarly get $\Delta^{\f}, \Gamma \vdash_c^{\e\f\p} A$, and we conclude using the fact that $\Delta^{\f}$ is derivable from functional extensionality in classical logic.
\end{proof}

\section{Conclusion}

Brown and Rizkallah~\cite{brown_rizkallah} extended Kuroda's translation to higher-order logic, showing that the soundness property holds, except in the presence of functional extensionality.
We have discussed several conditions under which the soundness property holds in the presence of functional extensionality, including the double-negation shift.

Brown and Rizkallah neither examined the characterization property nor the reverse translation property. 
We have shown that both properties do not hold in higher-order logic, although they are straightforward in first-order logic.
We have identified conditions under which such properties hold. 
In particular, it is sufficient to assume functional extensionality and propositional extensionality.  

We have seen that the role of both functional extensionality and propositional extensionality is predominant when extending Kuroda's translation to higher-order logic. 
One has to be careful when using both principles in an intuitionistic context: the Diaconescu-Goodman-Myhill theorem~\cite{diaconescu,goodman_myhill} states that, together with the axiom of choice, they entail the principle of excluded middle.

\section*{Acknowledgments} The author would like to thank Marc Aiguier, Gilles Dowek, Olivier Hermant and the anonymous reviewer for their helpful remarks about this work, and Dominik Kirst for pointing out that the double-negation elimination on equality predicates may entail the double-negation elimination.

\bibliographystyle{alpha}
\bibliography{biblio}

\end{document}